\documentclass[11pt]{article}

\usepackage[a4paper,margin=1in]{geometry}
\usepackage{amsmath,amssymb,amsthm}
\usepackage{bm}
\usepackage{authblk}
\usepackage{cite}
\usepackage{setspace}
\usepackage[colorlinks=true,linkcolor=blue,citecolor=blue,urlcolor=blue]{hyperref}

\newtheorem{theorem}{Theorem}

\title{Constraints on Reversing the Thermodynamic Arrow of Time\\
from Black Hole Thermodynamics, Wormholes,\\
and Time-Symmetric Quantum Mechanics}

\author[1]{Kevin Song}
\author[1]{John Zhang}
\affil[1]{\small Department of Biomedical Engineering\protect\\ The University of Alabama at Birmingham\protect\\ Birmingham, AL 35294, USA}

\date{December 2, 2025}

\begin{document}

\maketitle

\begin{abstract}
Can the thermodynamic arrow of time in a single universe be reversed, even temporarily, within semiclassical gravity without invoking additional universes or branches? We address this question in a single, connected spacetime where quantum field theory is coupled to classical general relativity, and where black holes, traversable wormholes, and time-symmetric or retrocausal formulations of quantum mechanics might naively appear to open channels for entropy export or cancellation. After distinguishing fine-grained, coarse-grained, and generalized gravitational entropy, and formulating a cosmological coarse-grained entropy, we treat black hole evaporation, wormholes constrained by quantum energy inequalities, and two-time boundary-value frameworks (including absorber-type and two-state-vector formalisms) within a common information-theoretic language. We then introduce a ``Global Entropy Transport'' (GET) framework and derive a sectoral inequality that bounds the net decrease of matter-plus-radiation entropy in terms of changes in horizon area and correlation (mutual-information) terms, assuming the generalized second law and modern focusing and energy conditions. Within this framework, black holes, wormholes, and retrocausal protocols can at most redistribute entropy among matter, radiation, and gravitational sectors and reshape the local pattern of entropy production. They do not, under current semiclassical, holographic, and statistical-mechanical constraints, permit a genuine reversal of the universal thermodynamic arrow in a single connected universe.
\end{abstract}

\newpage

\doublespacing

\section{Introduction}

The second law of thermodynamics, generalized to gravitational systems, states that appropriate entropy functionals for isolated systems, and for spacetimes with horizons, are nondecreasing under physically allowed processes.
In a cosmological context, this motivates the question: can the thermodynamic arrow of time for the universe be reversed, even for a finite interval, without appealing to additional universes or branches?

In a semiclassical regime where quantum field theory (QFT) is coupled to classical general relativity (GR), black hole thermodynamics and the generalized second law (GSL) provide strong evidence that a suitably defined generalized entropy cannot decrease \cite{Bekenstein1974,Hawking1975,Bardeen1973,Wald1999,Wall2012,Bousso2016}.
At the same time, solutions of Einstein's equations that contain wormholes or closed timelike curves, together with retrocausal or time-symmetric interpretations of quantum mechanics, suggest qualitatively new channels for energy and information flow \cite{MorrisThorne1988,MorrisThorneYurtsever1988,XianZhao2020,WheelerFeynman1945,WheelerFeynman1949,Cramer1986,Aharonov1964,EvansSearles2002}.
It is therefore natural to ask whether any of these structures, consistently embedded in a single connected universe, can be used to reduce a physically meaningful universal entropy.

The purpose of this paper is to formulate this question precisely and to show why the answer appears to be negative under broad but explicit assumptions.
We restrict attention to a single globally hyperbolic spacetime solving the semiclassical Einstein equations with suitable energy conditions in the averaged quantum sense, and to unitary QFT on that background.
We do not assume any ensemble of universes or Many-Worlds branching; our aim is to understand whether, given a fixed universe, clever use of black holes, wormholes, or retrocausal protocols could lower its thermodynamic entropy.

To do that, we first clarify which notion of entropy is relevant.
Fine-grained von Neumann entropy is conserved in closed-unitary evolution and cannot decrease or increase.
Thermodynamic discussions instead invoke coarse-grained entropy, defined by grouping microstates into macrostates determined by experimentally accessible variables.
In GR, horizons motivate a further generalized entropy which combines horizon area with the entropy of quantum fields outside, and which in holographic setups is computed by quantum extremal surfaces \cite{Wald1999,EngelhardtWall2015}.

We then analyze three classes of structures in a single-universe setting: black holes and horizon thermodynamics governed by the four laws of black hole mechanics and Hawking radiation; traversable wormholes and their stress-energy, constrained by quantum inequalities and modern quantum energy conditions; and time-symmetric quantum mechanical frameworks, including Wheeler--Feynman absorber theory, the Transactional Interpretation, and the Aharonov--Bergmann--Lebowitz (ABL) two-state formalism.
Building on these ingredients, we propose a ``Global Entropy Transport'' (GET) picture in which entropy is shuffled between matter, radiation, and gravitational sectors via nonlocal correlations and advanced or retarded field structures.
Our central technical result is a sectoral inequality (Theorem~\ref{thm:GET}), proved in Sec.~\ref{sec:GET}, that quantifies how much entropy can be removed from non-gravitational sectors without violating the GSL.

Our discussion is complementary to Past Hypothesis accounts of the cosmological arrow of time, in which a special low-entropy condition at (or near) the Big Bang underwrites the observed thermodynamic arrow \cite{Albert2000,Price1996,CarrollChen2004,Carroll2010,Penrose1979}.
We do not seek to explain that boundary condition.
Instead, we assume such a low-entropy past and ask what an internal agent can or cannot accomplish dynamically afterward, once black hole thermodynamics and semiclassical gravity are taken seriously.

Throughout, we use units with $c = k_{\mathrm B} = 1$ and write $\ell_{\mathrm P}^2 = G\hbar$.

\subsection{Assumptions and regime of validity}

Our conclusions are conditional on a standard set of semiclassical and holographic assumptions:

\begin{itemize}
  \item A single globally hyperbolic spacetime solving the semiclassical Einstein equations with suitable averaged quantum energy conditions.
  \item Unitary QFT for matter and radiation on that background; no fundamental non-unitary collapse dynamics.
  \item Validity of the GSL for the generalized entropy of relevant horizons, and of the Quantum Focusing Conjecture (QFC) and Quantum Null Energy Condition (QNEC) in the regimes where we invoke them \cite{Wall2012,Bousso2016,BoussoQNEC2016}.
  \item In holographic discussions, applicability of the quantum extremal surface (QES) prescription and the associated island formula to the setups under consideration \cite{EngelhardtWall2015,Almheiri2019,Almheiri2020,Penington2020}.
\end{itemize}

Within this framework we treat backreaction and horizon areas semiclassically, and use generalized entropy and QES technology as proxies for the underlying microscopic gravitational entropy.

\section{Entropy in an isolated semiclassical universe}

\subsection{Fine-grained and coarse-grained entropy}

Consider a closed quantum system with Hilbert space $\mathcal{H}$ and density operator $\rho(t)$ evolving unitarily under some Hamiltonian $H$.
The fine-grained von Neumann entropy
\begin{equation}
  S_{\mathrm{vN}}[\rho] = - \operatorname{Tr}\left( \rho \ln \rho \right)
\end{equation}
is conserved in time:
\begin{equation}
  \frac{\mathrm{d}}{\mathrm{d}t} S_{\mathrm{vN}}[\rho(t)] = 0,
\end{equation}
because the spectrum of $\rho$ is invariant under unitary evolution.

Irreversibility arises only after coarse-graining.
Let $\{P_\alpha\}$ be a set of mutually orthogonal projectors that encode experimentally accessible macroscopic observables, such that $\sum_{\alpha} P_\alpha = \mathbb{I}$ and $P_\alpha P_\beta = \delta_{\alpha\beta} P_\alpha$.
The coarse-grained state is
\begin{equation}
  \rho_{\mathrm{cg}} = \sum_{\alpha} P_\alpha \rho P_\alpha,
\end{equation}
and the associated coarse-grained entropy
\begin{equation}
  S_{\mathrm{cg}} = - \operatorname{Tr}\left( \rho_{\mathrm{cg}} \ln \rho_{\mathrm{cg}} \right)
\end{equation}
can increase with time when the microscopic dynamics mixes microstates within each macrostate, effectively spreading the support of $\rho$ over a larger number of macrostates.

In classical statistical mechanics, this corresponds to Boltzmann's entropy
\begin{equation}
  S_{\mathrm{B}} = \ln W,
\end{equation}
where $W$ is the phase space volume of microstates compatible with a given macrostate.
Liouville's theorem guarantees that the phase space volume of the full fine-grained distribution is invariant.
Coarse-graining suppresses small-scale structure, allowing the coarse-grained entropy to rise.

\subsection{Coarse-grained thermodynamic entropy for a cosmological spacetime}

For a spatially homogeneous and isotropic universe described by a Friedmann--Robertson--Walker (FRW) metric
\begin{equation}
  \mathrm{d}s^2 = - \mathrm{d}t^2 + a(t)^2 \left( \frac{\mathrm{d}r^2}{1 - k r^2} + r^2 \mathrm{d}\Omega_2^2 \right),
\end{equation}
with scale factor $a(t)$ and curvature parameter $k = 0, \pm 1$, we can define a coarse-grained thermodynamic entropy
\begin{equation}
  S_{\mathrm{th}}(t) = \sum_i s_i(T_i(t),\mu_i(t)) V_i(t),
\end{equation}
where $s_i$ is the entropy density of each component $i$ (radiation, baryonic matter, dark matter, etc.) as a function of local temperature $T_i$ and chemical potentials $\mu_i$, and $V_i$ is the comoving volume occupied by that component.

Gravitational degrees of freedom are not captured by this sum.
For de~Sitter-like accelerating phases, one associates an entropy to the cosmological horizon,
\begin{equation}
  S_{\mathrm{dS}} = \frac{A_{\mathrm{dS}}}{4 G \hbar},
\end{equation}
in direct analogy with black hole entropy \cite{GibbonsHawking1977}.
Penrose has argued that the low entropy of the early universe is tied to the smallness of the Weyl curvature and the absence of gravitational clumping, a ``gravitational Past Hypothesis'' \cite{Penrose1979}.
More recent estimates of the universe's entropy budget highlight the dominance of black holes and cosmological horizons as repositories of gravitational entropy \cite{EganLineweaver2010}.

In the presence of horizons, a more robust quantity than $S_{\mathrm{th}}$ is the generalized entropy
\begin{equation}
  S_{\mathrm{gen}} = \frac{A}{4 G \hbar} + S_{\mathrm{out}},
\end{equation}
where $A$ is the area of a causal horizon cross-section and $S_{\mathrm{out}}$ is the von Neumann entropy of quantum fields outside the horizon \cite{Bekenstein1974,Wald1999,Wall2012,EngelhardtWall2015}.
We take $S_{\mathrm{gen}}$ as the primary candidate for a universal thermodynamic entropy in a semiclassical universe, and later connect it to quantum extremal surfaces and island constructions.

\section{Black hole thermodynamics and the generalized second law}

\subsection{Four laws and Bekenstein--Hawking entropy}

For stationary black holes in GR, the four laws of black hole mechanics mirror the ordinary laws of thermodynamics \cite{Bardeen1973,Wald1997Review}.

Let $M$ be the mass, $J$ the angular momentum, $Q$ the charge, $\kappa$ the surface gravity of the horizon, $A$ the horizon area, and $\Omega_{\mathrm{H}}$, $\Phi_{\mathrm{H}}$ the angular velocity and electric potential at the horizon.
The first law reads
\begin{equation}
  \delta M = \frac{\kappa}{8 \pi G} \delta A + \Omega_{\mathrm{H}} \, \delta J + \Phi_{\mathrm{H}} \, \delta Q.
\end{equation}
The zeroth law states that $\kappa$ is constant over the event horizon of a stationary black hole.
The second law (Hawking area theorem) states that in any classical process satisfying the null energy condition,
\begin{equation}
  \delta A \geq 0,
\end{equation}
and the third law forbids reaching $\kappa = 0$ in a finite advanced time.

Hawking's calculation that black holes radiate thermally with temperature \cite{Hawking1975}
\begin{equation}
  T_{\mathrm{H}} = \frac{\hbar \kappa}{2 \pi}
\end{equation}
motivates the identification of black hole entropy as
\begin{equation}
  S_{\mathrm{BH}} = \frac{A}{4 G \hbar}.
\end{equation}
This identification fits the first law in the form $\delta M = T_{\mathrm{H}} \delta S_{\mathrm{BH}} + \cdots$ and is supported by microstate counting in several quantum gravity models.

\subsection{Generalized second law and quantum focusing}

Once quantum fields are included, Hawking radiation can decrease $A$, leading to apparent violations of the classical area theorem.
Bekenstein proposed the generalized second law (GSL) \cite{Bekenstein1974}:
\begin{equation}
  \Delta S_{\mathrm{BH}} + \Delta S_{\mathrm{out}} \geq 0,
\end{equation}
where $S_{\mathrm{out}}$ is the entropy outside the horizon.
Equivalently, the generalized entropy
\begin{equation}
  S_{\mathrm{gen}} = \frac{A}{4 G \hbar} + S_{\mathrm{out}}
\end{equation}
is nondecreasing along appropriate slices of the horizon.

Modern proofs of the GSL in semiclassical regimes, for a wide class of horizons and matter couplings, have been given by Wall and others \cite{Wall2012,Wall2010}.
In parallel, Bousso and collaborators formulated the Quantum Focusing Conjecture (QFC) \cite{Bousso2016}, which introduces a quantum generalization of the expansion of null congruences.
Let $\sigma$ be a codimension-two surface, $N$ a null hypersurface orthogonal to $\sigma$, and $S_{\mathrm{gen}}(\lambda)$ the generalized entropy of a cut of $N$ at affine parameter $\lambda$.
The quantum expansion is defined as
\begin{equation}
  \Theta(\lambda) = \frac{4 G \hbar}{\sqrt{h}} \frac{\delta S_{\mathrm{gen}}}{\delta \lambda},
\end{equation}
where $\sqrt{h}$ is the determinant of the induced metric on $\sigma$.
The QFC asserts that
\begin{equation}
  \frac{\mathrm{d}\Theta}{\mathrm{d}\lambda} \leq 0,
\end{equation}
which in many situations implies a local form of the GSL and bounds the way generalized entropy can evolve along null surfaces.

Related progress on quantum energy conditions, such as the quantum null energy condition (QNEC), further constrains allowed stress tensors and helps to link entropy inequalities to local energy densities \cite{BoussoQNEC2016}.
Taken together, these results strongly support the statement that any physically admissible process in semiclassical gravity obeys
\begin{equation}
  \Delta S_{\mathrm{gen}} \geq 0
\end{equation}
for suitable choices of horizons and slicing.
Any attempt to decrease universal entropy in a controlled way must therefore either violate the GSL, invalidate semiclassical approximations, or rely on very special boundary conditions.

\section{Traversable wormholes, negative energy, and chronology protection}

\subsection{Morris--Thorne wormholes and energy conditions}

A static, spherically symmetric traversable wormhole can be modeled by the Morris--Thorne metric \cite{MorrisThorne1988}
\begin{equation}
  \mathrm{d}s^2 = - e^{2\Phi(r)} \mathrm{d}t^2 + \frac{\mathrm{d}r^2}{1 - b(r)/r} + r^2 \mathrm{d}\Omega_2^2,
\end{equation}
where $\Phi(r)$ is a redshift function and $b(r)$ is a shape function satisfying $b(r_0) = r_0$ at the throat radius $r_0$.
To avoid horizons, $\Phi(r)$ must be finite everywhere.
The flaring-out condition at the throat requires $b'(r_0) < 1$.

Using Einstein's equations $G_{\mu\nu} = 8\pi G T_{\mu\nu}$, one finds that the radial null energy condition (NEC),
\begin{equation}
  T_{\mu\nu} k^\mu k^\nu \geq 0
\end{equation}
for all null vectors $k^\mu$, is violated near the throat \cite{MorrisThorne1988,MorrisThorneYurtsever1988}.
For example, choosing a radial null vector $k^\mu$ with components $(1,\sqrt{1 - b(r)/r},0,0)$, one obtains at the throat
\begin{equation}
  T_{\mu\nu} k^\mu k^\nu \propto \frac{b'(r_0) - 1}{r_0^2} < 0.
\end{equation}
Maintaining a macroscopic traversable wormhole therefore requires ``exotic'' matter with negative energy density in some frame.

\subsection{Quantum inequalities and constraints on macroscopic wormholes}

Quantum field theory allows negative energy densities in certain states, such as squeezed vacuum and Casimir configurations, but these are constrained by quantum inequalities (QI) that limit the magnitude and duration of negative energy along timelike or null worldlines \cite{FordReview2009,FordQICurved1998}.
A typical QI has the schematic form
\begin{equation}
  \int_{-\infty}^{\infty} \mathrm{d}\tau\, f(\tau) \langle T_{\mu\nu} u^\mu u^\nu \rangle \geq - \frac{C}{\tau_0^4},
\end{equation}
where $\tau$ is proper time along a timelike geodesic with tangent $u^\mu$, $f(\tau)$ is a sampling function of width $\tau_0$, and $C$ is a positive constant depending on the field content.
Such bounds imply that sustaining a large, macroscopic region of negative energy for long times is extremely difficult.

Applied to wormhole geometries, QI bounds and related quantum energy conditions suggest that either the throat radius is Planckian or the amount of exotic matter must be enormous and arranged with extreme precision \cite{FordQICurved1998,BoussoQNEC2016}.
This severely constrains the viability of traversable wormholes engineered as macroscopic thermodynamic devices.

\subsection{Chronology protection and wormhole heat flow}

If a traversable wormhole connects two regions of spacetime with a relative time shift, it can generate closed timelike curves (CTCs).
Hawking's chronology protection conjecture \cite{Hawking1992} asserts that quantum effects will conspire to prevent the formation of macroscopic CTCs, for example by causing stress-energy to diverge near the would-be Cauchy horizon.

More recently, Xian and Zhao studied heat flow through a wormhole connecting two black holes in an AdS$_2$-dual setup, starting from a thermofield double state.
They found that, while a wormhole channel can produce anomalous heat flow from colder to hotter side at the cost of consuming entanglement, the net thermodynamic arrow is preserved once all channels are included \cite{XianZhao2020}.
In particular, the wormhole does not generate a macroscopic reversal of the second law, but rather redistributes correlations between the two sides.
This is naturally interpreted as a realization of Global Entropy Transport: subsystem entropies and apparent heat flows can be nonstandard, but once horizon and entanglement contributions are included the total generalized entropy remains nondecreasing.

These results, together with QI bounds, imply that wormholes cannot be used as ``entropy drains'' that export entropy away from an otherwise isolated universe to lower its generalized entropy.
At best, they provide a nontrivial way to re-route energy and entanglement while respecting the GSL.

\section{Time-symmetric quantum mechanics and entropy production}

\subsection{Wheeler--Feynman absorber theory and advanced fields}

The Wheeler--Feynman absorber theory reformulates classical electrodynamics in terms of direct interparticle action, eliminating free fields and treating the electromagnetic field as a bookkeeping device \cite{WheelerFeynman1945,WheelerFeynman1949}.
Each charged particle emits a half-advanced, half-retarded field; the presence of a complete absorber in the future enforces boundary conditions such that the net observed field is purely retarded.

In modern notation, the total field at spacetime point $(\mathbf{x},t)$ can be written schematically as
\begin{equation}
  A^{\mu}_{\mathrm{tot}}(x) = \frac{1}{2} \sum_n \left( A^{\mu}_{n,\mathrm{ret}}(x) + A^{\mu}_{n,\mathrm{adv}}(x) \right),
\end{equation}
with the absorber condition arranged so that the homogeneous solution
\begin{equation}
  A^{\mu}_{\mathrm{free}}(x) = \frac{1}{2} \sum_n \left( A^{\mu}_{n,\mathrm{ret}}(x) - A^{\mu}_{n,\mathrm{adv}}(x) \right)
\end{equation}
vanishes.
The resulting effective dynamics is time-asymmetric at the level of observed macroscopic fields, even though the underlying equations are time-symmetric.

From a thermodynamic perspective, the absorber condition is a special final boundary condition that selects retarded behavior and aligns the electromagnetic arrow of time with the thermodynamic arrow, rather than reversing it.
The microscopic phase space or Hilbert space evolution remains measure-preserving, so coarse-grained entropy still tends to increase in the usual way.

\subsection{Transactional Interpretation and two-state vector formalism}

Cramer's Transactional Interpretation (TI) recasts quantum processes as ``transactions'' between retarded offer waves and advanced confirmation waves, forming a time-symmetric handshake between emitters and absorbers \cite{Cramer1986}.
Mathematically, this can be related to the two-state vector formalism (TSVF) of Aharonov, Bergmann, and Lebowitz \cite{Aharonov1964}, where a system between two measurements is described by a pair of state vectors, one evolving forward in time from the past boundary and one backward in time from the future boundary.

In TSVF, the probability for an outcome at an intermediate time is given in terms of both the pre-selected state $|\psi_i\rangle$ and the post-selected state $|\psi_f\rangle$.
The ABL rule for a nondegenerate measurement of observable $A$ with projectors $\{P_a\}$ is
\begin{equation}
  \Pr(a | \psi_i, \psi_f) = \frac{|\langle \psi_f | P_a | \psi_i \rangle|^2}{\sum_b |\langle \psi_f | P_b | \psi_i \rangle|^2}.
\end{equation}
Despite its time-symmetric form, the underlying unitary dynamics is still governed by the Schr\"odinger equation.
Entropy production in macroscopic systems, derived from coarse-graining and typicality, is not altered by the reinterpretation of amplitudes.

\subsection{Fluctuation theorems and the arrow of time}

Fluctuation theorems quantify the probability of entropy-decreasing trajectories in small systems over finite times \cite{EvansSearles2002}.
For entropy production rate $\overline{\Sigma}_t$ averaged over time $t$, a canonical form of the fluctuation theorem is
\begin{equation}
  \frac{\Pr\left( \overline{\Sigma}_t = A \right)}{\Pr\left( \overline{\Sigma}_t = -A \right)} = e^{A t},
\end{equation}
for suitable definitions of $\overline{\Sigma}_t$ in classical or quantum systems.
For example, in a driven system initially in equilibrium at temperature $T$, the Evans--Searles or Crooks relations give, schematically,
\begin{equation}
  \frac{P(\Delta S_{\mathrm{tot}} = \sigma)}{P(\Delta S_{\mathrm{tot}} = -\sigma)}
  = e^{\sigma},
\end{equation}
where $\Delta S_{\mathrm{tot}}$ is the total entropy production (in units of $k_{\mathrm B}$) over a protocol.
Entropy-decreasing fluctuations ($\sigma < 0$) are exponentially suppressed for large systems and long times, reproducing the second law in the thermodynamic limit.

Time-symmetric interpretations like TI or TSVF do not modify these relations, since they are derived from the underlying unitary dynamics and the choice of initial ensembles.
Retrocausal narratives may offer alternative stories about how particular outcomes are coordinated by both past and future boundary conditions, but they do not furnish a mechanism to globally bias the statistics toward entropy-decreasing histories without altering the boundary conditions themselves.
In particular, no reinterpretation of amplitudes can change the fluctuation-theorem ratio without changing the microscopic dynamics or the ensemble preparation.

\section{Global Entropy Transport in a single universe}
\label{sec:GET}

\subsection{Conceptual setup}

We now introduce a single-universe framework for what we call Global Entropy Transport (GET).
The idea is to formalize processes in which entropy is shuffled among sectors in a way that might appear, to some observers, as partial reversal of the thermodynamic arrow in a subsystem, while the generalized entropy of appropriate horizons still increases.

Let the total Hilbert space factorize as
\begin{equation}
  \mathcal{H} = \mathcal{H}_{\mathrm{grav}} \otimes \mathcal{H}_{\mathrm{m}} \otimes \mathcal{H}_{\mathrm{rad}},
\end{equation}
where $\mathcal{H}_{\mathrm{grav}}$ encodes geometric and horizon degrees of freedom, and $\mathcal{H}_{\mathrm{m}},\mathcal{H}_{\mathrm{rad}}$ encode matter and radiation.
In reality, such a tensor product factorization is subtle in quantum gravity, but it is a useful semiclassical approximation.

The fine-grained entropy of the total state remains constant:
\begin{equation}
  S_{\mathrm{vN}}[\rho_{\mathrm{tot}}] = \text{const}.
\end{equation}
We define subsystem entropies $S_{\mathrm{grav}}, S_{\mathrm{m}}, S_{\mathrm{rad}}$ by tracing out complementary factors, and a generalized entropy
\begin{equation}
  S_{\mathrm{gen}} = \frac{A}{4 G \hbar} + S_{\mathrm{out}}[\rho_{\mathrm{m}},\rho_{\mathrm{rad}}],
\end{equation}
where $A$ is the area of the relevant horizon and $S_{\mathrm{out}}$ is the entropy of fields outside the horizon.

It is convenient to write
\begin{equation}
  S_{\mathrm{out}} = S_{\mathrm{m}} + S_{\mathrm{rad}} + S_{\mathrm{corr}},
\end{equation}
where $S_{\mathrm{corr}}$ encodes correlations and entanglement between matter and radiation.
Equivalently,
\begin{equation}
  S_{\mathrm{corr}} = S_{\mathrm{out}} - S_{\mathrm{m}} - S_{\mathrm{rad}} = - I(\mathrm{m}:\mathrm{rad}) \le 0,
\end{equation}
with $I(\mathrm{m}:\mathrm{rad})$ the mutual information between matter and radiation.

Global Entropy Transport is any process in which the entropies of accessible non-gravitational sectors decrease,
\begin{align}
  \Delta S_{\mathrm{m}} &< 0, \\
  \Delta S_{\mathrm{rad}} &< 0,
\end{align}
in some region or time interval, while the generalized entropy still obeys the GSL,
\begin{equation}
  \Delta S_{\mathrm{gen}} \geq 0.
\end{equation}
Local or sectoral entropy may then decrease, but only at the expense of increased entanglement with gravitational degrees of freedom and a compensating increase of horizon area or of $S_{\mathrm{corr}}$ elsewhere.

\subsection{Black holes as entropy sinks and sources}

Black holes provide a concrete realization of GET.
Dropping a low-entropy object into a black hole increases its area and thus $S_{\mathrm{BH}}$, while reducing the entropy in the exterior region.
The GSL ensures that
\begin{equation}
  \Delta S_{\mathrm{BH}} + \Delta S_{\mathrm{out}} \geq 0.
\end{equation}
From the perspective of an exterior observer who cannot access interior degrees of freedom, the black hole functions as an entropy sink.

Conversely, Hawking radiation gradually transfers entropy back to the exterior.
For a Schwarzschild black hole of mass $M$, the Hawking temperature is
\begin{equation}
  T_{\mathrm{H}} = \frac{\hbar}{8 \pi G M},
\end{equation}
and the power radiated scales as $P \sim A T_{\mathrm{H}}^4 \sim M^{-2}$.
The total entropy radiated during complete evaporation is of order $M^2$ in Planck units, comparable to the initial $S_{\mathrm{BH}}$.
Thus, from a global perspective, the black hole is both a sink and a source, mediating redistribution of entropy between geometric and field sectors rather than destroying it.

\subsection{Wormhole-mediated energy and entropy transport}

In a hypothetical traversable wormhole connecting two distant regions, matter and radiation can be transported between the mouths in ways that may appear to reverse local entropy gradients.
However, as Xian and Zhao demonstrated in a controlled AdS$_2$ setup, anomalous heat flow from colder to hotter subsystems is accompanied by depletion of entanglement and constrained by the overall dynamics \cite{XianZhao2020}.
Moreover, QI bounds and chronology protection restrict the macroscopic controllability of such wormholes.

Formally, one can consider an effective two-reservoir model with temperatures $T_1$ and $T_2$ and heat currents
\begin{equation}
  J_{\mathrm{tot}} = J_{\mathrm{diff}} + J_{\mathrm{anom}},
\end{equation}
where $J_{\mathrm{diff}} \propto (T_1 - T_2)$ is the ordinary diffusive term and $J_{\mathrm{anom}}$ encodes wormhole-mediated anomalous flow.
The entropy production rate is
\begin{equation}
  \dot{S}_{\mathrm{prod}} = J_{\mathrm{tot}} \left( \frac{1}{T_2} - \frac{1}{T_1} \right).
\end{equation}
Even if $J_{\mathrm{anom}}$ temporarily dominates and drives $J_{\mathrm{tot}}$ opposite to the usual sign, the full backreaction and entanglement dynamics must be considered to evaluate $\dot{S}_{\mathrm{prod}}$.
The GSL suggests that, once geometric contributions are included, the total generalized entropy remains nondecreasing.
This wormhole-driven two-reservoir setup is thus a second concrete realization of GET beyond the black-hole engine.

\subsection{Retrocausality and two-time boundary conditions}

A more radical route to entropy manipulation would use two-time boundary conditions.
Suppose the universe is described by a path integral with both an initial and final density matrix, $\rho_{\mathrm{i}}$ and $\rho_{\mathrm{f}}$.
The probability for a coarse-grained history $\alpha$ is then determined by both boundary conditions.
In TSVF language, the universe is described by a pair $(|\Psi_{\mathrm{i}}\rangle, |\Psi_{\mathrm{f}}\rangle)$.

If $\rho_{\mathrm{f}}$ is highly non-generic and very low-entropy, it can in principle bias the statistics of intermediate histories toward lower-entropy configurations.
However, this does not violate unitarity; it simply means that the actual universe has very special final conditions, in addition to special initial conditions.
Such a model can be written schematically as
\begin{equation}
  P[\alpha] \propto \operatorname{Tr} \left( C_\alpha \rho_{\mathrm{i}} C_\alpha^\dagger \rho_{\mathrm{f}} \right),
\end{equation}
where $C_\alpha$ is a class operator for history $\alpha$.
Unless $\rho_{\mathrm{f}}$ is chosen in an ad hoc way, this does not provide an operational protocol that agents inside the universe can exploit to reverse entropy.
Instead, it reinterprets the arrow of time as arising from boundary conditions at both ends.

In any case, once one conditions on both boundaries, the GSL for generalized entropy along appropriate horizons remains a local statement about the allowed evolution of $S_{\mathrm{gen}}$ between cuts.
Time-symmetric or retrocausal formulations therefore do not, by themselves, provide a mechanism for global entropy decrease.
They relocate the explanatory burden from dynamics to boundary conditions.

\subsection{A worked example of GET: matter--black-hole engine}

To make GET concrete, consider an internal agent operating a quasi-static engine between a macroscopic matter reservoir at temperature $T_{\mathrm{m}}$ and a Schwarzschild black hole with Hawking temperature $T_{\mathrm{H}}$.
We assume $T_{\mathrm{m}} > T_{\mathrm{H}}$, which is generic for astrophysical black holes.

Let a small amount of heat $\delta Q$ be extracted from the matter reservoir and dumped into the black hole, with negligible work cost and negligible change in $T_{\mathrm{m}}$ and $T_{\mathrm{H}}$.
The entropy changes are
\begin{align}
  \Delta S_{\mathrm{m}} &= - \frac{\delta Q}{T_{\mathrm{m}}}, \\
  \Delta S_{\mathrm{BH}} &= \frac{\delta Q}{T_{\mathrm{H}}},
\end{align}
where the second relation follows from the first law $\delta M = T_{\mathrm{H}} \delta S_{\mathrm{BH}}$ with $\delta M = \delta Q$ in our units.

The generalized entropy change for the relevant horizon is then
\begin{equation}
  \Delta S_{\mathrm{gen}} = \Delta S_{\mathrm{BH}} + \Delta S_{\mathrm{out}}
  = \frac{\delta Q}{T_{\mathrm{H}}} - \frac{\delta Q}{T_{\mathrm{m}}}
  = \delta Q \left( \frac{1}{T_{\mathrm{H}}} - \frac{1}{T_{\mathrm{m}}} \right).
\end{equation}
For positive $\delta Q$ and $T_{\mathrm{m}} > T_{\mathrm{H}}$, we have $\Delta S_{\mathrm{gen}} > 0$.
At the same time, the matter sector experiences a genuine entropy decrease $\Delta S_{\mathrm{m}} < 0$.
This is a simple instance of GET: an internal agent has ``mined'' entropy from the matter reservoir and stored it in the black hole, with a strict increase in the generalized entropy.

If one attempts to reverse the roles, arranging for net heat flow from the colder black hole to the hotter matter reservoir, one must supply external work.
Even in that case, the Clausius inequality applied to the combined system of reservoirs and engine implies
\begin{equation}
  \Delta S_{\mathrm{BH}} + \Delta S_{\mathrm{m}} \geq 0,
\end{equation}
so the generalized entropy cannot be reduced by such an engine.
The detailed accounting simply becomes more complicated as one includes the work reservoir and radiation, but GET still describes a redistribution of entropy rather than a global reduction.

\subsection{A sectoral inequality for GET}

Using the decomposition
\begin{equation}
  S_{\mathrm{gen}} = \frac{A}{4 G \hbar} + S_{\mathrm{m}} + S_{\mathrm{rad}} + S_{\mathrm{corr}},
\end{equation}
a general GET process yields the identity
\begin{equation}
  \Delta S_{\mathrm{gen}} =
  \Delta\!\left(\frac{A}{4 G \hbar}\right)
  + \Delta S_{\mathrm{m}}
  + \Delta S_{\mathrm{rad}}
  + \Delta S_{\mathrm{corr}}.
\end{equation}
Imposing the GSL,
\begin{equation}
  \Delta S_{\mathrm{gen}} \geq 0,
\end{equation}
we obtain
\begin{equation}
  \Delta S_{\mathrm{m}} + \Delta S_{\mathrm{rad}}
  \geq
  - \Delta\!\left(\frac{A}{4 G \hbar}\right) - \Delta S_{\mathrm{corr}}.
  \label{eq:GETbound}
\end{equation}

\begin{theorem}[Global Entropy Transport bound]
\label{thm:GET}
Assume a single semiclassical spacetime satisfying the GSL for the generalized entropy $S_{\mathrm{gen}}$ of the relevant horizons, with matter and radiation described by unitary QFT on that background.
For any finite process that begins and ends in semiclassical states for which $S_{\mathrm{gen}}$ is well-defined, the change in the combined entropy of matter and radiation satisfies the bound
\begin{equation}
  \Delta S_{\mathrm{m}} + \Delta S_{\mathrm{rad}}
  \geq
  - \Delta\!\left(\frac{A}{4 G \hbar}\right) - \Delta S_{\mathrm{corr}},
\end{equation}
where $S_{\mathrm{corr}} = -I(\mathrm{m}:\mathrm{rad})$ is minus the mutual information between matter and radiation.
\end{theorem}

\begin{proof}
By definition,
\begin{equation}
  S_{\mathrm{gen}} = \frac{A}{4 G \hbar} + S_{\mathrm{m}} + S_{\mathrm{rad}} + S_{\mathrm{corr}}.
\end{equation}
Taking the difference between initial and final states gives
\begin{equation}
  \Delta S_{\mathrm{gen}} =
  \Delta\!\left(\frac{A}{4 G \hbar}\right)
  + \Delta S_{\mathrm{m}}
  + \Delta S_{\mathrm{rad}}
  + \Delta S_{\mathrm{corr}}.
\end{equation}
The GSL asserts $\Delta S_{\mathrm{gen}} \geq 0$ for the horizons in question, so rearranging yields the claimed inequality.
\end{proof}

Equation~\eqref{eq:GETbound} is therefore a simple but useful constraint: any attempt to engineer a large net reduction in the combined entropy of matter and radiation must be accompanied either by a compensating increase in horizon area (a growth of gravitational entropy) or by a compensating change in correlations encoded in $S_{\mathrm{corr}}$.

Subadditivity and strong subadditivity of entropy restrict how negative $S_{\mathrm{corr}}$ can become.
In particular, for two subsystems one has the mutual information
\begin{equation}
  0 \leq I(\mathrm{m}:\mathrm{rad}) = S_{\mathrm{m}} + S_{\mathrm{rad}} - S_{\mathrm{out}},
\end{equation}
and hence $S_{\mathrm{corr}} = S_{\mathrm{out}} - S_{\mathrm{m}} - S_{\mathrm{rad}} = - I(\mathrm{m}:\mathrm{rad}) \leq 0$.
Thus $\Delta S_{\mathrm{corr}}$ is bounded below by $-2\min(S_{\mathrm{m}},S_{\mathrm{rad}})$, and large changes require creating or destroying significant mutual information.
In practice, arranging a large negative $\Delta S_{\mathrm{corr}}$ in a controlled way is extremely difficult, since it requires coherent manipulation of highly entangled degrees of freedom while respecting QNEC and QFC constraints \cite{Bousso2016,BoussoQNEC2016}.

Theorem~\ref{thm:GET} therefore functions as a no-go result for ``entropy bleaching'' devices that operate entirely within the matter and radiation sectors.
Even exotic setups employing wormholes or advanced fields must pay the cost encoded in the change of horizon area and correlations.

\subsection{Operational limitations for internal agents}

An internal agent in our universe might imagine various devices aiming to globally reduce entropy:

\begin{itemize}
  \item \emph{Entropy-mining wormholes}: use a traversable wormhole to export high-entropy radiation from a local region into some inaccessible sector, while importing low-entropy matter.
  \item \emph{Advanced-signal refrigerators}: exploit advanced fields or retrocausal correlations to construct Maxwell-demon-like feedback protocols that coordinate microscopic degrees of freedom in the past.
  \item \emph{Black-hole erasers}: feed arbitrary data into black holes and attempt to recover a purified, low-entropy radiation state using sophisticated decoding and island-engineering techniques.
\end{itemize}

In each case, the GET framework combined with the GSL and Theorem~\ref{thm:GET} constrains what can be achieved.
Wormhole-based devices must respect quantum inequalities and QNEC, which severely limit sustained negative energy densities and exotic stress tensors.
Advanced-signal refrigerators cannot alter the global unitary evolution and therefore merely reshuffle entropy between the agent, the controlled system, and the environment; the total generalized entropy remains bounded below by its initial value.
Black-hole erasers run into the fact that, once islands are taken into account, any purification of Hawking radiation is balanced by corresponding entropy changes in the gravitational sector, as we discuss next.

From the standpoint of an internal experimenter, GET says that one can mine entropy from some subsystems and export it into others, including horizon degrees of freedom, but cannot reduce the generalized entropy defined on appropriate quantum extremal surfaces.
The arrow of time can be locally sharpened or blurred, yet it cannot be globally reversed by any operational protocol consistent with semiclassical gravity and quantum information inequalities.

\section{Holographic generalized entropy, islands, and the Page curve}

\subsection{Quantum extremal surfaces and the island formula}

In holographic theories, the fine-grained entropy of a boundary region $A$ is computed by extremizing a generalized entropy functional over bulk surfaces homologous to $A$.
The original Ryu--Takayanagi (RT) proposal for static setups \cite{RyuTakayanagi2006} and its covariant Hubeny--Rangamani--Takayanagi (HRT) extension \cite{HubenyRangamaniTakayanagi2007} associate to $A$ a bulk extremal surface $\gamma_A$ with area $A(\gamma_A)$, such that
\begin{equation}
  S(A) = \frac{A(\gamma_A)}{4 G \hbar}.
\end{equation}
Engelhardt and Wall generalized this to \emph{quantum extremal surfaces} (QES), where one extremizes the generalized entropy
\begin{equation}
  S_{\mathrm{gen}}(\Sigma) = \frac{A(\Sigma)}{4 G \hbar} + S_{\mathrm{bulk}}[\Sigma]
\end{equation}
over codimension-two surfaces $\Sigma$ and selects the minimal extremum \cite{EngelhardtWall2015}.
Here $S_{\mathrm{bulk}}[\Sigma]$ is the von Neumann entropy of quantum fields in the bulk region outside $\Sigma$.

In evaporating black hole setups coupled to an external bath, these ideas lead to the island formula for the entropy of Hawking radiation.
Denoting by $R$ the radiation region and by $I$ a putative island inside the black hole, one finds
\begin{equation}
  S(R) = \min_{\partial I} \; \mathrm{ext}_{\partial I}
  \left[
    \frac{A(\partial I)}{4 G \hbar} + S_{\mathrm{bulk}}(R \cup I)
  \right],
\end{equation}
where the minimization runs over QES candidates $\partial I$ \cite{Almheiri2019,Almheiri2020,Penington2020}.
The quantity inside the brackets is precisely a generalized entropy of the form we have used throughout.
The appearance of islands ensures that $S(R)$ follows the unitary Page curve: it rises initially as Hawking quanta accumulate, then turns over at the Page time when islands become dominant and begins to decrease.

Most explicit computations of the Page curve and islands have been performed in AdS or near-AdS settings, often with lower-dimensional toy models.
Extending a fully microscopic QES-based description of entropy to realistic cosmological spacetimes is an active area of research.
Here we extrapolate the generalized-entropy intuition provided by QES and islands to semiclassical cosmological settings, using it as a guide rather than as a rigorously established microscopic derivation.

\subsection{GET, the Page curve, and non-decrease of generalized entropy}

The island formula and the Page curve clarify how fine-grained entropy and generalized entropy coexist.
The fine-grained entropy of the combined system of black hole plus radiation is constant under unitary evolution, while the entropy of the radiation subsystem $S(R)$ rises and then falls.
The generalized entropy associated with the relevant QES surfaces, however, remains nondecreasing along appropriate null foliations, in harmony with the GSL and the QFC \cite{Bousso2016,EngelhardtWall2015,Almheiri2019}.

Our GET framework is naturally interpreted in this language.
The gravitational sector in Sec.~\ref{sec:GET} can be thought of as including QES degrees of freedom and island regions, while $S_{\mathrm{m}}$ and $S_{\mathrm{rad}}$ capture coarse-grained matter and radiation entropies in the bath.
Theorem~\ref{thm:GET} then constrains what an internal agent can do to $S_{\mathrm{m}}$ and $S_{\mathrm{rad}}$ without violating the monotonicity of the generalized entropy evaluated on QES surfaces.

In particular, schemes that attempt to engineer a Page-curve-like decrease of radiation entropy earlier than dictated by the usual island transition would require either manipulating the location of QES surfaces or altering $S_{\mathrm{bulk}}$ in a way that still respects QNEC and QFC.
There is no indication in current holographic calculations that such manipulations can lead to a net reduction of $S_{\mathrm{gen}}$; instead, they change which QES dominates the entropy calculation without overthrowing the underlying generalized entropy monotonicity \cite{Almheiri2019,Almheiri2020,Penington2020}.

\section{Cosmological arrows, Past Hypothesis, and outlook}

Past Hypothesis frameworks locate the origin of the thermodynamic arrow in a special low-entropy macrostate of the early universe, often formalized as a constraint on the initial coarse-grained distribution over microstates \cite{Albert2000,Price1996,CarrollChen2004,Carroll2010,Penrose1979}.
The Weyl curvature hypothesis \cite{Penrose1979} and subsequent work on gravitational entropy \cite{GibbonsHawking1977,EganLineweaver2010} emphasize that the absence of gravitational clumping at early times plays a central role in making that low-entropy state extraordinarily special.

Our analysis does not challenge these accounts.
Rather, it addresses a different question: given such a low-entropy beginning and the subsequent dynamics described by semiclassical gravity and QFT, what can internal agents hope to accomplish in terms of reversing or significantly reducing entropy thereafter?

Within the GET framework, the answer appears to be restrictive.
Once we take seriously the GSL, quantum focusing, quantum energy conditions, and holographic generalized entropy, any dynamical process available to internal agents can only redistribute entropy among sectors.
Black holes, wormholes, and retrocausal correlations permit intricate patterns of Global Entropy Transport, in which $S_{\mathrm{m}}$ and $S_{\mathrm{rad}}$ can be locally reduced at the expense of increasing horizon area and correlation terms.
However, the generalized entropy $S_{\mathrm{gen}}$, evaluated on appropriate horizons or QES surfaces, is constrained to be nondecreasing.

In this sense, our results complement Past Hypothesis explanations.
The Past Hypothesis explains why a thermodynamic arrow exists at all: it posits a special low-entropy initial macrostate.
The GET framework and Theorem~\ref{thm:GET} instead quantify the robustness of that arrow against manipulation by internal agents, even ones equipped with black-hole engineering, wormholes, and sophisticated quantum information processing.
They can generate pockets of low entropy, export entropy to gravitational sectors, and cleverly manage correlations, but they cannot globally roll back the generalized entropy inherited from the initial low-entropy boundary condition.

Future progress in quantum gravity and cosmology may refine these conclusions.
A full microscopic understanding of gravitational entropy, including in non-asymptotically AdS cosmologies, could tighten the link between QES-based generalized entropy and cosmological arrows.
Conversely, novel quantum energy conditions or nonperturbative effects might reveal new channels for entropy transport that are invisible in current semiclassical treatments.
As things stand, however, the GSL, QFC, QNEC, and island-based Page curve computations together provide strong evidence that the universal thermodynamic arrow of time in a single connected universe cannot be dynamically reversed.

\begin{singlespace}

\end{singlespace}

\end{document}